\documentclass[12pt]{amsart}
\usepackage{amsmath, amsthm, latexsym, amssymb, amsfonts, epsf, xcolor, hyperref}
\usepackage{caption, subcaption, enumerate, color}
\usepackage{multirow,multicol}
\usepackage{mathtools}
\usepackage{bm}
\usepackage{blkarray}
\usepackage{booktabs}
\usepackage{framed}
\usepackage{setspace}
\usepackage{tikz}
\usepackage{graphicx}
\usepackage{wrapfig}
\usepackage{float}
\usepackage[dvips]{epsfig}
\usepackage{bbm}
\usepackage[utf8]{inputenc}
\usepackage{url}
\usepackage{booktabs}
\usepackage{bibentry}
\usepackage[normalem]{ulem}
\usepackage{comment}
\usepackage{lingmacros}
\usepackage{cleveref}

\usepackage{extarrows}
\usepackage{multirow}
\usepackage{mathtools}
\usepackage{enumerate}
\usepackage{array}
\usepackage{bigstrut}
\usepackage{exscale,relsize}
\usepackage{graphicx}

\usepackage[linesnumbered]{algorithm2e}
\usetikzlibrary{shapes.geometric}

\newcommand{\cal}[1]{\mathcal{#1}}

\newcommand{\cL}{\cal L}
\newcommand{\cM}{\cal M}

\newcommand{\cX}{\cal X}

\newcommand{\F}{{\mathbb F}}

\newcommand{\Z}{{\mathbb Z}}

\newcommand{\X}{{\mathcal X}}

\def \F{{\mathbb F}}
\def \Z{{\mathbb Z}}

\def \lm{\textrm{lm}}
\def\M{\mathcal{M}}  
\def\fq{\mathbb{F}_q}
\def\fqr{\mathbb{F}_{q^r}}

\newtheorem{theorem}{Theorem}[section]
\newtheorem{lemma}[theorem]{Lemma}
\newtheorem{proposition}[theorem]{Proposition}
\newtheorem{corollary}[theorem]{Corollary}

\theoremstyle{definition}
\newtheorem{definition}[theorem]{Definition} 
\newtheorem{remark}[theorem]{Remark}
\newtheorem{example}[theorem]{Example}

\theoremstyle{plain}

\newcommand{\hull}{\operatorname{Hull}}

\newcommand{\rmv}[1]{}

\begin{document}


\title[Decreasing norm-trace codes]{Decreasing norm-trace codes}
\author{C\'icero Carvalho}
\address[C\'icero Carvalho]{Faculdade de Matem\'{a}tica \\ Universidade Federal 
de Uberl\^{a}ndia\\ Uberl\^{a}ndia, MG Brazil}
\email{cicero@ufu.br}

\author{Hiram H. L\'opez}
\address[Hiram H. L\'opez]{Department of Mathematics\\ Cleveland State University\\ Cleveland, OH USA}
\email{h.lopezvaldez@csuohio.edu}

\author{Gretchen L. Matthews}
\address[Gretchen L. Matthews]{Department of Mathematics\\ Virginia Tech\\ Blacksburg, VA USA}
\email{gmatthews@vt.edu}
\thanks{The first author was partially supported by FAPEMIG APQ-00864-21. The second author was 
partially supported by NSF DMS-2201094. The third author was partially 
supported by NSF DMS-2201075 and the Commonwealth Cyber Initiative.}
\keywords{Evaluation codes, decreasing monomial codes, norm-trace curves, repair scheme, dual codes.}
\subjclass[2010]{94B05; 11T71; 14G50}

\begin{abstract}
The decreasing norm-trace codes are evaluation codes defined by a set of monomials closed under divisibility and the rational points of the extended norm-trace curve. In particular, the decreasing norm-trace codes contain the one-point algebraic geometry (AG) codes over the extended norm-trace curve. We use Gr\"obner basis theory and find the indicator functions on the rational points of the curve to determine the basic parameters of the decreasing norm-trace codes: length,  dimension, and minimum distance. We also obtain their dual codes. We give conditions for a decreasing norm-trace code to be a self-orthogonal or a self-dual code. We provide a linear exact repair scheme to correct single erasures for decreasing norm-trace codes, which applies to higher rate codes than the scheme developed by Jin, Luo, and Xing (IEEE Transactions on Information Theory {\bf 64} (2), 900-908, 2018) when applied to the one-point AG codes over the extended norm-trace curve.
\end{abstract}

\maketitle

\section{Introduction}
Decreasing monomial codes, which are evaluation codes in which the set of monomials is closed under divisibility, were introduced by Bardet, Dragoi, Otmani, and Tillich in \cite{bardet2016} to algebraically analyze the polar codes defined by Arikan \cite{arikan}. The families of monomials satisfying a closure property also appeared in an earlier construction for optimized evaluation codes \cite{BO}.

The classical families of Reed-Solomon and Reed-Muller codes are decreasing monomial codes and have amply motivated the study of wider classes of decreasing monomial codes.  Decreasing monomial-Cartesian codes, also known as variants of Reed-Muller codes over finite grids, are  more general families than Reed-Muller codes that have been studied due to their applications to certain symmetric channels~\cite{polar_decreasing}, distributed storage systems~\cite{lmv}, and efficient decoding algorithms~\cite{STV}.

In~\cite{Eduardo_polar}, 
Camps, Mart\'{i}nez-Moro, and Rosales
 introduced Vardøhus codes, which are polar codes defined by kernels from castle curves. They proved that Vardøhus codes are polar codes for a discrete memoryless channel that is symmetric with respect to the field operations. \rmv{The authors also studied the dual and provided a minimum distance bound for the family of decreasing castle codes.}

In this paper, we study decreasing norm-trace codes, which are decreasing monomial codes where the evaluation points are the rational points of the extended norm-trace curve. As a consequence of~\cite{Eduardo_polar} and the fact that the extended norm-trace curve is Castle, the decreasing norm-trace codes can be considered as polar codes for discrete memoryless channels symmetric with respect to the field operations.

Let $\fqr$ be the finite field with $q^r$ elements. Norm-trace codes are defined using the {\it norm-trace curve} which is an affine curve over $\F_{q^r}$ defined by \[ N(x)=Tr(y) \]
where $N(x)$ is the norm and $Tr(y)$ is the trace, both taken with respect to the 
extension $\F_{q^r}/\F_q$. Let $u$ be a positive integer such that $u \mid \frac{q^r - 1}{q - 1}$. The {\it extended norm-trace curve}, denoted by $\cX_u$, is the affine curve over $\F_{q^r}$ defined by the equation
\[
x^{u} = y^{q^{r - 1}} + y^{q^{r - 2}} + \cdots + y.
\]
We focus in this work on decreasing norm-trace codes, which are codes defined by evaluating monomials on the rational points of the curve $\cX_u$. We now give more details.

Enumerate the rational points on $\cX_u$ so that $\cX_u = \left\{P_1,\ldots,P_n \right\} \subseteq \fqr^2$.
The \textit{evaluation map}, denoted ${\rm ev}$, is the $\fqr$-linear map given by  
$$
\begin{array}{lccc}
{\rm ev}\colon &\fqr[x,y] &\rightarrow& \fqr^{n}\quad \\
&f & \mapsto& \left(f(P_1),\ldots,f(P_n)\right).
\end{array}
$$
Let $\mathcal{M} \subseteq \fqr[x,y]$ be a set of monomials closed under 
divisibility, meaning that if $M\in \mathcal{M}$ and $M^\prime$ divides $M$, 
then $M^\prime \in \mathcal{M}$. Let $\mathcal{L}$ be the 
$\fqr$-subspace of $\fqr[x,y]$ generated by the set $\mathcal{M}$. We call the image of 
$\mathcal{L}$ under the evaluation map, denoted by ${\rm ev}(\mathcal{M})$, a {\it decreasing norm-trace code}. 
We can see that the extended norm-trace codes introduced and recently studied 
in \cite{bras-amor} and \cite{Heera-Pin} are particular instances of decreasing 
norm-trace codes; norm-trace codes and generalizations also appear in \cite{BO}. Moreover, we check later that the family of decreasing 
norm-trace codes contains, as a specific case, the family of one-point 
geometric Goppa codes over the Hermitian curve and the more general norm-trace 
curve.

We organize this paper as follows. In Section~\ref{preli}, we describe the vanishing ideal $I_{\cX_u}$ of the extended norm-trace curve $\cX_u$  (Lemma~\ref{nts}), which is the ideal of all polynomials that vanish on $\cX_u$. We recall essential concepts from the Gr\"obner basis theory, such as the footprint of an ideal, and determine a Gr\"obner basis for $I_{\cX_u}$ (Proposition~\ref{22.06.03}) with respect to the lexicographic order.

The main result of Section~\ref{standard} shows the standard indicator function of every rational point of the extended norm-trace curve $\cX_u$ (Theorem~\ref{22.06.05}). Given a rational point $P$ on $\cX_u$, a standard indicator function $f_P$ is a linear combination of monomials that belong to the footprint of $I_{\cX_u}$ such that $f_P(P)=1$ and $f_P(P^\prime)=0$ for every other rational point $P^\prime \neq P$ of $\cX_u$. 

In Section~\ref{decreasing}, we formally introduce decreasing norm-trace codes (Definition~\ref{22.03.11}). We determine their basic parameters, such as the length, dimension, and minimum distance (Theorem~\ref{22.06.13}). We show that these decreasing monomial codes generalize the one-point AG codes over the norm-trace curve.

In Section~\ref{22.09.14}, we give an explicit expression for the dual of a decreasing norm-trace code (Theorem \ref{22.03.15}) in terms of the complement of the {set of} monomials. The hull of a linear code is the intersection of the code with its dual. The hull has several applications, e.g., it has been used to classify finite projective planes~\cite{assmus} and to construct entanglement-assisted quantum error-correcting codes~\cite{guenda}. We show instances where the hull of a decreasing norm-trace code is computed explicitly (Theorem \ref{22.06.12}). We also give conditions on the set of monomials, so that the decreasing norm-trace code is a self-orthogonal or a self-dual code.

In Section~\ref{erasure}, we apply our results to study linear repair schemes for decreasing norm-trace codes. A repair scheme is an
algorithm that recovers the value at any entry of a codeword using limited information from 
the values at the other entries. After presenting the basic definitions of this theory, we prove results that show the existence of a repair scheme for decreasing norm-trace codes (Theorem~\ref{22.06.08}). We close with some conclusions at the end of Section~\ref{conclusion}.

References for vanishing ideals and related algebraic concepts used in this work are \cite{CLO1, Eisen, harris, monalg}.
 
\section{Preliminaries}\label{preli}

Let $\fq$ be the finite field with $q$ elements and $r \geq 2$ an integer. Define
the polynomials $N(x):=x^{\frac{q^r-1}{q-1}}$ and $Tr(y):=y^{q^{r - 1}} + 
y^{q^{r - 2}} + \cdots + y^q + y$ in $\fqr[x,y]$. The {\it trace} with respect to the 
extension $\F_{q^r}/\F_q$  is  the map
$$
\begin{array}{lccc}
Tr: & \F_{q^r} & \to & \F_q \\
& \alpha & \mapsto & Tr(\alpha).
\end{array}
$$
The {\it norm} with respect to the extension $\F_{q^r}/\F_q$  is  the 
map
$$
\begin{array}{lccc}
N: & \F_{q^r} & \rightarrow & \F_q \\
& \alpha & \mapsto & N(\alpha).
\end{array}
$$
The {\it norm-trace curve}, denoted by $\mathcal{X}$, is the affine plane curve over $\fqr$ given by the equation 
\[
x^{\frac{q^r - 1}{q - 1}} = y^{q^{r - 1}} + y^{q^{r - 2}} + \cdots + y.
\]
The curve $\mathcal{X}$ has been extensively studied in the literature to construct linear codes~\cite{bras-amor, geil, kim_boran, nt_lifted, MTT_08}. We are interested in a slightly more general curve. Let $u$ be a positive integer such that $u \mid \frac{q^r - 1}{q - 1}$.
The {\it extended norm-trace curve}, denoted by $\cX_u$, is the affine curve over $\F_{q^r}$ defined by the equation
\[
x^{u} = y^{q^{r - 1}} + y^{q^{r - 2}} + \cdots + y.
\]
We use the rational points of the curve $\cX_u$ to construct a family of decreasing evaluation codes, which contains, as a particular case, the extended norm-trace codes~\cite{bras-amor, Heera-Pin}. \rmv{ Before introducing the new family of evaluation codes, we need better to understand the curve $\cX_u$ and its properties.}

Results from  Gr\"{o}bner bases theory have been used in coding theory for some time to determine parameters of codes (see, e.g., \cite{geil, geil2, carvalho2015}). We now recall important concepts and results from this theory.

Let $\M$ be the set of monomials of $\fqr[x_1, \ldots, x_m]$. A {\em monomial order} $\prec$ on $\M$ is a total order where 1 is the least monomial and if $M_1 \prec M_2$, then $M M_1 \prec M M_2$, for all $M, M_1, M_2 \in \M$. Fix a monomial order in $\M$ and let $f$ be a nonzero polynomial in $\fqr[x_1,\ldots, x_m]$. The greatest monomial which 
appears in $f$ is called the {\em leading monomial} of $f$, denoted by $\lm(f)$.
Given an ideal $I \subseteq \fqr[x_1, \ldots, x_m]$, a {\em Gr\"{o}bner basis} for $I$ is a set $\{f_1, \ldots, f_s\} 
\subseteq I$ such that for every polynomial $f \in I \setminus \{0\}$, we have that $\lm(f)$ is a multiple of $\lm(f_i)$ for some $i \in \{1, \ldots, s\}$. The Gr\"{o}bner basis concept was introduced in the Ph.D. thesis of Bruno Buchberger (see \cite{buchberger}), in which the author proves that every ideal admits a Gr\"{o}bner basis (w.r.t.\ a fixed monomial order) and that if $\{f_1, \ldots, f_s\}$ is a Gr\"{o}bner basis for $I$, then $I = (f_1, \ldots, f_s)$.

Let $\{f_1, \ldots, f_s\}$ be a Gr\"{o}bner basis for $I$. The {\em footprint} of $I$ is the set $\Delta_{\prec}(I)$ of monomials which are not multiples of $\lm(f_i)$ for all $i = 1, \ldots, n$. One of the main results in 
Buchberger's thesis states that the set of classes $\{ M + I \mid M \in \Delta_{\prec}(I) \} 
\subseteq \fqr[x_1, \ldots, x_m]/I$ is a basis for $\fqr[x_1, \ldots, x_m]/I$ as a $\fqr$-vector space.  

We now define an ideal associated with the extended norm-trace curve $\cX_u$:
$$I_{\cX_u} := ( Tr(y) - x^u, x^{q^r} - x, y^{q^r} - y ) \subseteq \fqr[x,y].$$
Next, we consider some relevant properties of this ideal. 

\begin{lemma} \label{nts}
The ideal $I_{\cX_u}$ is radical and is the ideal of all polynomials which vanish on $\cX_u$.
\end{lemma}
\begin{proof}
Since $x^{q^r} - x, y^{q^r} - y \in I_{\cX_u}$,  for any monomial order, the 
footprint is finite and consists of monomials of the form $x^a y ^b$ where 
$a$ and $b$ are less than $q^r$. Hence, $I_{\cX_u} $ is a 
zero-dimensional ideal. Thus, $I_{\cX_u} $ is a radical ideal by \cite[Prop. 8.14]{becker}. From \cite[Thm. 2.3]{ghorpade}, it follows that $I_{\cX_u}$ is the ideal 
of all polynomials which vanish on $\cX_u$. The fact that $I_{\cX_u}$ is radical also follows immediately from \cite[Thm. 2.3]{ghorpade} being a vanishing ideal. 
\end{proof}

The following result is a particular case of \cite[Theorem 21]{Heera-Pin}. We add a detailed proof here for completeness.
\begin{proposition}\label{22.06.03}
The set $\{ Tr(y) - x^u, x^{(q - 1)u + 1} - x\}$ is a Gr\"obner basis 
for $I_{\cX_u}$ with respect to the lexicographic order with $x \prec y$. Moreover, $\mid \X_u\mid = q^{r-1}((q - 1)u + 1)$. In particular, if $u = \frac{q^r - 1}{q - 1}$, then $I_{\cX} $ is the vanishing ideal of $\mathcal{X}$, the set
$\{Tr(y)-N(x),x^{q^r} - x\}$
is a Gr\"obner basis for $I_{\cX} $ with respect to the lexicographic order with $x \prec y$,  and $\mid \X \mid = q^{2r-1}$.
\end{proposition}
\begin{proof}
Let $(\alpha,\beta)$ be a point on $\X_u$. As $\alpha^u = Tr(\beta)$,  $\alpha^u \in \mathbb{F}_q$. 
Thus, the polynomial $x^{(q - 1)u + 1} - x = x( (x^u)^{q - 1} - 1)$ vanishes at 
all points of $\X_u$ and from Lemma \ref{nts},  $x^{(q - 1)u + 1} - x \in I_{\cX_u}$.
To prove that $I_{\cX_u}  = ( Tr(y) - x^u, x^{(q - 1)u + 1} - x )$, we show that $x^{q^r} 
- x, y^{q^r} - y \in ( Tr(y) - x^u, x^{(q - 1)u + 1} - x )$. Indeed, let $v$ be the 
positive integer such that $u v = \frac{q^r - 1}{q - 1}$. Then one easily checks 
that
\begin{eqnarray*}
&&\left(x^{(q - 1)u(v-1)} + x^{(q - 1)u(v-2)} + \cdots + x^{(q - 1)u} + 1 
\right) ( x^{(q - 1)u + 1} - x) = x^{q^r} - x \\
\text{and} && \left( (Tr(y) - x^u)^{q - 1} - 1\right) (Tr(y) - x^u) + x^{u - 1}(x^{(q - 
1)u + 1} - x)  \\
&& = (Tr(y) - x^u)^q - (Tr(y) - x^u) + x^{uq} - x^u \\ 
&& = Tr(y)^q - x^{uq} - Tr(y) + x^u + x^{uq} - x^u \\
&& = (y^{q^{r - 1}} + y^{q^{r - 2}} + \cdots + y)^q - (y^{q^{r - 1}} + y^{q^{r 
- 2}} + \cdots + y) \\
&& = y^{q^r}  - y.
\end{eqnarray*}
Since the leading monomials $\lm(Tr(y)) = y^{q^{r - 1}}$ and $\lm(x^{(q - 1)u + 1} - x) = 
x^{(q - 1)u + 1}$ are coprime,  $\{ Tr(y) - x^u, x^{(q - 1)u + 1} - x\}$ 
is a Gr\"obner basis for $I_{\cX_u}$ according to 
\cite[Prop. 4, p. 104]{CLO1}.  Since $I_{\cX_u} $ is a 
radical ideal and $\fqr$ is a perfect field, $ \mid \X_u \mid = \mid 
\Delta_{\prec}(I_{\cX_u} ) \mid$  
\cite[Thm. 8.32]{becker}.
\end{proof}

\rmv{
Observe that we have the following result when $u = \frac{q^r - 1}{q - 1}$ in Proposition~\ref{22.06.03}.
\begin{corollary}
If $\cX$ is the norm-trace curve $x^{\frac{q^r - 1}{q - 1}} = y^{q^{r - 1}} + y^{q^{r - 2}} + \cdots + y$,
and $I_{\cX} $ is the vanishing ideal of $\mathcal{X}$, then the set
$\{Tr(y)-N(x),x^{q^r} - x\}$
is a Gr\"obner basis for $I_{\cX} $ with respect to the lexicographic order with $x \prec y$. Moreover, $\mid \X \mid = q^{2r-1}$.
\end{corollary}}

\section{Standard indicator functions}\label{standard}
Some of the properties of the decreasing evaluation codes depend on the indicator functions of the curve $\X_u$. For a standard reference on indicator functions in the context of evaluation codes, please see \cite[Sections 4 and 5]{dual}. Take $n := \mid \X_u \mid$. One may show (see \cite{flax} or \cite[Prop. 3.7]{carvalho2015}) that the following linear transformation is an isomorphism
\[ 
\begin{array}{lccc}
\varphi \colon &\fqr[x, y]/I_{\X_u} &\rightarrow& \fqr^{n}\quad \\
&f + I_{\X_u} & \mapsto& \left(f(P_1),\ldots,f(P_n)\right).
\end{array}
\]
So, for each $P \in \X_u$, there exists an unique class $g_P + I_{\X_u}$ such that $g_P(P) = 1$ and $g_P(Q) = 0$, for 
every $Q \in \X_u \setminus \{ P \}$. Since $\{ M + I_{\cX_u} \mid M \in \Delta_{\prec}\left(I_{\cX_u}\right) \}$ is a 
basis for $\fqr[x,y]/I_{\X_u}$ as an $\fqr$-vector space (see e.g.\ \cite[Prop. 6.52]{becker}), there is a unique $\fqr$-linear combination  of monomials in $\Delta_{\prec}(I_{\cX_u})$, which we denote by $f_P$, such that $f_P(P) = 1$ and $f_P(Q) = 0$, for every $Q \in \X_u$. We call this polynomial $f_P$ the {\em standard indicator function of $P$}. The existence and uniqueness of the standard indicator function $f_P$ may also be seen as a consequence of \cite[Lemma 4.2(c)]{dual}. 

We first describe the set of points of $\cX_u$  in a way that will be useful in the next section.
\begin{lemma}\label{points}
For every $\gamma \in \fq$, define
$A_\gamma := \{ (\alpha,\beta) \in \mathbb{A}^2(\fqr) \, | \, Tr (\beta) = \alpha^u = \gamma\}$. Then, we have
$\cX_u = \bigcup_{\gamma \in \fq} A_\gamma$. Moreover, $\mid A_0\mid = q^{r - 1}$ and 
$\mid A_\gamma\mid = u q^{r - 1}$ for all $\gamma \in \fq^*$.
\end{lemma}
\begin{proof}
If $(\alpha, \beta) \in \cX_u$, then $\alpha$ is a root of $x^{(q - 1)u + 1} - x$ by Proposition~\ref{22.06.03}.
Furthermore,  $x^{(q - 1)u + 1} - x \mid x^{q^r} - x$, so
$x^{(q - 1)u + 1} - x$ has $(q - 1)u + 1$ distinct roots. For each nonzero root $\alpha$, we have $(\alpha^u)^{q- 1} = 1$,
so $\alpha^u \in \fq^*$ and $x^{(q - 1)u + 1} - x$ must be a factor of $ x \prod_{\gamma \in \fq^*}  (x^u - \gamma)$.
Since the last two polynomials have the same degree and are monic,
we actually have $x^{(q - 1)u + 1} - x  = x \prod_{\gamma \in \fq^*} (x^u - \gamma)$.

There are $q^{r - 1}$ elements $\beta$ such that $Tr(\beta) = \alpha^u$ for every $\gamma \in \fq$.
This shows that $A_\gamma \subseteq \cX_u$ for all $\gamma \in \fq$, and that 
$\mid A_0\mid = q^{r - 1}$ and $\mid A_\gamma\mid = u q^{r - 1}$ when $\gamma \in \fq^*$. 
On the other hand,  if $(\alpha, \beta) \in 
\cX_u$ and $\alpha^u = \gamma$, then $(\alpha, \beta) \in A_\gamma$.
\end{proof}

\begin{theorem}\label{22.06.05}
Let $P =(\alpha,\beta) \in \X_u$. The polynomial 
\[ f_P(x,y):= c \left(\frac{ x^{(q - 1)u + 1} - x }{x-\alpha}\right) \left(\frac{Tr(y)-Tr(\beta)}{y-\beta}\right) \]
is the standard indicator function for $P$ where $$c := \begin{cases} -1 & \textnormal{if } \alpha=0\\(-u)^{-1} \in \fq & \textnormal{otherwise}. \end{cases}$$
In particular, $y^{q^{r - 1}-1}x^{(q - 1)u}$ is the leading monomial of the standard indicator function for $P$. 
\end{theorem}
\begin{proof}
Observe that $f_P(\alpha,\beta) \neq 0$.
Let $P^\prime=\left(\alpha^\prime,\beta^\prime\right)$ be a point in $\cX_u$ 
different from $P$. If $\alpha\neq \alpha^\prime$, then $\left(\frac{ x^{(q - 1)u 
+ 1} - x }{x-\alpha}\right)\bigg\rvert_{x=\alpha^\prime}=0$ by 
Proposition~\ref{22.06.03}; thus $f_P(P^\prime)=0$. 
If $\alpha = \alpha^\prime$, then $\beta\neq \beta^\prime$ and 
$Tr(\beta^\prime)=(\alpha^\prime)^u=\alpha^u=Tr(\beta)$. This means that
$\beta^\prime$ is a root of $Tr(y)-Tr(\beta)$ and $ 
\left(\frac{Tr(y)-Tr(\beta)}{y-\beta}\right)\bigg\rvert_{y=\beta^\prime}=0$; 
thus $f_P(P^\prime)=0$. 
We conclude that $f_P(\alpha',\beta') = 0$ for every $(\alpha', \beta') \in \X_u \setminus \{(\alpha,\beta)\}$.

Let $\beta = \beta_1,\ldots,\beta_{q^{r-1}} \in  \mathbb{F}_{q^r}$ be the distinct roots of 
$Tr(y)-Tr(\beta)$, so that 
$Tr(y)-Tr(\beta)=\prod_{i=1}^{q^{r-1}}(y-\beta_i)$. Taking the formal derivative, we have that
$1=\sum_{j=1}^{q^{r-1}} \prod_{i=1, i\neq j}^{q^{r-1}} (y-\beta_i)$. Thus, 
\[
\left(\frac{Tr(y)-Tr(\beta)}{y-\beta}\right)\bigg\rvert_{y=\beta} = 
\prod_{i= 2}^{q^{r-1}} (\beta -\beta_i) = 1.
\]
Likewise, let $\alpha = \alpha_1, \ldots, \alpha_{u(q - 1) + 1}$ be such that $x^{(q - 1)u + 
1} - x = \prod_{i=1}^{(q - 1)u + 1}(x-\alpha_i)$. Taking the formal derivative, we 
get $( (q - 1)u + 1) x^{(q - 1) u} - 1 = \sum_{j=1}^{(q - 1)u + 1} \prod_{i=1, 
i\neq j}^{(q - 1)u} (x-\alpha_i)$, so 
\[
\left(\frac{ x^{(q - 1)u + 1} - x }{x-\alpha}\right)\bigg\rvert_{x = \alpha} = 
\prod_{i= 2}^{(q - 1)u} (\alpha -\alpha_i) = ( (q - 1)u + 1) \alpha^{(q - 1) u} - 1.
\]
If $\alpha \neq 0$, then we have $\alpha^u \in \fq^*$ from the proof of Lemma \ref{points}. Thus,
\[
( (q - 1)u + 1) \alpha^{(q - 1) u} - 1 = ( (q - 1)u + 1) (\alpha^{u})^{q - 1} - 1 = (q - 
1)u = -u.
\]
As $u \mid \frac{q^r - 1}{q - 1}$, the integer $u$ is not a multiple of $\textrm{char}(\fq)$. Hence $u \neq 0$ in $\fq$.
\end{proof}

\section{Code construction}\label{decreasing}
In this section, we define and compute the basic parameters of a new family of evaluation codes called decreasing norm-trace codes. We then consider their relationship with algebraic geometry codes defined on (extended) norm-trace curves

\subsection{Code parameters}

The \textit{evaluation map}, denoted ${\rm ev}$,
is the $\fqr$-linear map given by  
$$
\begin{array}{lccc}
{\rm ev}\colon &\fqr[x,y] &\rightarrow& \fqr^{n}\quad \\
&f & \mapsto& \left(f(P_1),\ldots,f(P_n)\right),
\end{array}
$$ 
where $\cX_u = \left\{P_1,\ldots,P_n \right\} \subseteq \fqr^2$ and $n:=q^{r-1}((q - 1)u + 1)$.
\rmv{Let $\mathcal{L}$ be a linear subspace of $\fqr[x,y]$ of finite dimension. The image of $\mathcal{L}$
under the evaluation map, denoted ${\rm ev}(\mathcal{L})$, is called an
\textit{evaluation code} on $\cX_u$. When $\mathcal{L}=\fqr\mathcal{M}$ is the subspace of $\fqr[x,y]$ generated by a set of
monomials $\mathcal{M} \subseteq \fqr[x,y]$, the evaluation code on $\cX_u$ is denoted by ${\rm ev}(\mathcal{M})$.}

\begin{definition}\label{22.03.11}
 A {\it decreasing norm-trace code} is an evaluation code ${\rm ev}(\mathcal{M})$ such that $\mathcal{M} \subseteq \fqr[x,y]$ is closed under divisibility, meaning  if $M\in \mathcal{M}$ and $M^\prime$ divides $M,$ then $M^\prime \in \mathcal{M}$.
\end{definition}

\begin{example}\label{22.03.10} \rm
Take $q=3$ and $r=2$. Figure~\ref{22.03.13}~(a) shows the points of the norm-trace curve $\cX$. Let $\textcolor{red}{\mathcal{M}}$ be the set of monomials in $\fqr[x,y]$ whose exponents are the points in Figure~\ref{22.03.13}~(b). Note that $\textcolor{red}{\mathcal{M}}$ is closed under divisibility. Using the coding theory package~\cite{cod_package} for Macaulay2 \cite{Mac2} and Magma \cite{magma}, we obtain that ${\rm ev}(\textcolor{red}{\mathcal{M}})$ is a $[ 27, 10, 15 ]$ decreasing norm-trace code over $\F_9$.
\end{example}
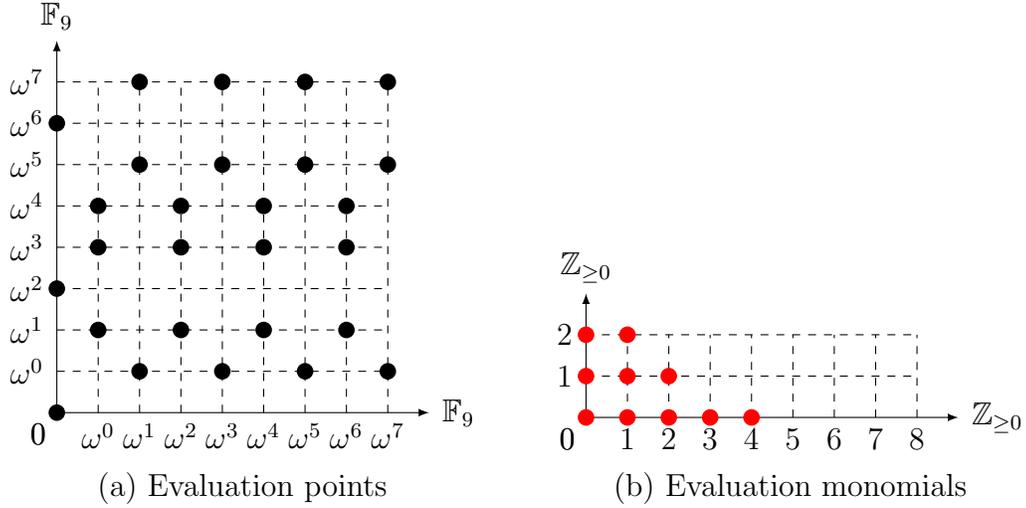
\begin{figure}[h]
\vskip 0cm
\noindent
\begin{minipage}[t]{0.45\textwidth}
\begin{center}
\begin{tikzpicture}[scale=0.55]
\def\xinitial{0}
\def\xfinal{7}
\def\yinitial{0}
\def\yfinal{7}

\draw [-latex] (\xinitial,\yinitial)node[below left]{0} -- (\xinitial,\yfinal+2)node[above]{$\F_{9}$};
\foreach \i in {\xinitial,...,\xfinal}
{
\draw [dashed] (\i+1,\yinitial)node[below]{$\omega^{\i}$} -- (\i+1,\yfinal+1)node[right] {};
}

\draw [-latex] (\xinitial,\yinitial)node[below left]{0} -- (\xfinal +2,\yinitial)node[right]{$\F_{9}$};
\foreach \i in {\yinitial,...,\yfinal}
{
\draw [dashed] (\xinitial,\i+1)node[left]{$\omega^{\i}$} -- (\xfinal+1,\i+1)node[right] {};
}

\foreach \j in {-1,2,6}
{
\fill [color=black](0,\j+1) {circle(.2cm)};
}

\foreach \i in {0,2,4,6}
{\foreach \j in {1,3,4}
{\fill [color=black](\i+1,\j+1) {circle(.2cm)};}}

\foreach \i in {1,3,5,7}
{\foreach \j in {0,5,7}
{\fill [color=black](\i+1,\j+1) {circle(.2cm)};}}

\end{tikzpicture}
\vskip 0cm
(a) Evaluation points
\end{center}
\end{minipage}
\begin{minipage}[t]{0.45\textwidth}
\begin{center}
\begin{tikzpicture}[scale=0.55]
\def\xinitial{1}
\def\xfinal{8}
\def\yinitial{1}
\def\yfinal{2}

\draw [-latex] (\xinitial,\yinitial)node[below left]{0} -- (\xinitial,\yfinal+2)node[above]{$\Z_{\geq 0}$};
\foreach \i in {\xinitial,...,\xfinal}
{
\draw [dashed] (\i+1,\yinitial)node[below]{$\i$} -- (\i+1,\yfinal+1)node[right] {};
}

\draw [-latex] (\xinitial,\yinitial)node[below left]{0} -- (\xfinal +2,\yinitial)node[right]{$\Z_{\geq 0}$};
\foreach \i in {\yinitial,...,\yfinal}
{
\draw [dashed] (\xinitial,\i+1)node[left]{$\i$} -- (\xfinal+1,\i+1)node[right] {};
}

\foreach \i in {0,1}
{\foreach \j in {0,...,2}
{\fill [color=red](\i+1,\j+1) {circle(.2cm)};}}

\foreach \i in {2}
{\foreach \j in {0,1}
{\fill [color=red](\i+1,\j+1) {circle(.2cm)};}}

\foreach \i in {3,4}
{\foreach \j in {0}
{\fill [color=red](\i+1,\j+1) {circle(.2cm)};}}

\end{tikzpicture}
\vskip 0cm
(b) Evaluation monomials
\end{center}
\end{minipage}
\caption{Take $q=3$ and $r=2$. (a) Shows the points of the norm-trace curve $\cX: x^4=y^3+y$. Let $\textcolor{red}{\mathcal{M}}$ be the set of monomials whose exponents are the points in (b). The evaluation code ${\rm ev}(\textcolor{red}{\mathcal{M}})$ is an $[ 27, 10, 15 ]$ decreasing norm-trace code over $\F_9$.}
\label{22.03.13}
\vskip 0.cm
\end{figure}

Denote by $\Delta\left(x^{(q-1)u+1}, y^{q^{r-1}}\right)$ the set of monomials that are not multiples of either of these two monomials.
%
%
%
From now on, we assume that $\mathcal{M} \subseteq \Delta\left(x^{(q-1)u+1}, 
y^{q^{r-1}}\right)$. 

We come to one of the main results of this work, which computes the basic parameters of a decreasing norm-trace code.
\begin{theorem}\label{22.06.13}
The decreasing norm-trace code ${\rm ev}(\mathcal{M})$ has the following basic parameters.
\begin{itemize}
\item[\rm (1)] Length $n = \mid \cX_u \mid = ((q - 1)u + 1)q^{r-1}$.
\item[\rm (2)] Dimension $k = \mid \mathcal{M} \mid$.
\item[\rm (3)]Minimum distance 
\begin{align*}
d = &((q - 1)u + 1)q^{r-1} \\ &- \max\left( \{ \min\left(a q^{r - 1} + 
(u(q-1) + 1 - a) b , \,  a q^{r - 1} + b u\right) \mid x^a y^b \in \cM  
\}\right).
\end{align*}
\end{itemize}
\end{theorem}
\begin{proof}
Statement (1) is a consequence of Proposition~\ref{22.06.03}. Statement (2) 
follows from the fact that $\{ M + I_{\cX_u} \mid M \in 
\Delta_{\prec}\left(I_{\cX_u}\right) \}$ is a 
basis for $\fqr[x,y]/I_{\X_u}$ as an $\fqr$-vector space together with the fact 
that the linear transformation $\varphi$, defined at the beginning of Section 
3, is an isomorphism. To prove Statement 
(3), let $f$ be a nonzero polynomial in the $\fqr$-vector space generated by 
the monomials in $\cM$. The 
set of points in $\cX_u$ which are zeros of $f$ is the set of the zeros of the 
ideal $I_{\cX_u} 
+ (f) \subseteq \fqr[x,y]$, denoted by $V(I_{\cX_u} 
+ (f))$.  The ideal  $I_{\cX_u} + (f)$ is a radical ideal; see \cite[Prop. 
8.14]{becker}. Therefore,  \cite[Thm. 8.32]{becker} implies that
$\mid V(I_{\cX_u}  + (f))\mid = \Delta_{\prec}(I_{\cX_u} + (f))$. 
Let $x^a y^b$ be the leading monomial of $f$ and 
let $\Delta\left(x^{(q-1)u+1}, y^{q^{r-1}}, x^a y^b \right)$ be the set of 
monomials that are not multiples of either of these three monomials. Then 
$\Delta_{\prec}(I_{\cX_u} + (f)) \subseteq \Delta(x^{(q-1)u+1}, y^{q^{r-1}}, x^a 
y^b)$ so 
that 
\begin{align*}
\mid V(I_{\cX_u}  + (f))\mid &\leq \mid\Delta(x^{(q-1)u+1}, y^{q^{r-1}}, x^a y^b)\mid \\ 
&= ((q-1)u + 1)q^{r - 1} - ((q-1)u + 1 - a)(q^{r-1} - b) \\
&= a q^{r-1} + ((q - 1)u + 1 - a)b.
\end{align*}
On the other hand, from \cite[Proposition 4]{geil2}), we have that 
$\mid V(I_{\cX_u}  + (f))\mid \leq a q^{r - 1}  +  b u$.

Assume that $a q^{r - 1}  +  b u 
\leq a q^{r-1} + ((q - 1)u + 1 - a)b$ and $b \neq 0$, so we have $a \leq (q - 
2) u + 1$. According to Lemma 
\ref{points} (and its proof), for all $\gamma \in fq^*$, the number 
of distinct elements of $\fqr$, which appear as the first entry of points in $A_\gamma$, is $u$, while $0$ is the first entry in all points of $A_0$. Fix $\gamma \in \fq^*$. Since 
$a \leq (q - 2) u + 1$, we may choose $\alpha_1, \ldots, \alpha_a \in \fqr$ such 
that for all $i = 1, \ldots, a$ we have $(\alpha_i, \beta_i) \in \cX_u$ for 
some $\beta_i \in \fqr$ and $\alpha_i^u \neq \gamma$. Recall that $b < q^{r 
- 1}$, and let $\beta_1, \ldots, \beta_b \in \fqr$ be distinct elements such 
that $Tr(\beta_j) = \gamma$ for all $j = 1, \ldots, b$.
Let $g(x,y) = \prod_{i = 1}^a (x - \alpha_i) \cdot \prod_{j = 1}^b (y - 
\beta_j)$. For every $i = 1, \ldots, a$, there exist $q^{r - 1}$ points in 
$\cX_u$ of the form $(\alpha_i, \beta)$, none of them in $A_\gamma$. For 
every $j = 1, \ldots, b$, there exist $u$ points of $\cX_u$ of the form 
$(\alpha, \beta_j)$, all of them in $A_\gamma$. Hence, $\mid V(I_{\cX_u}  + (g))\mid = a 
q^{r 
- 1}  +  b u$.

Now assume that $a q^{r-1} + ((q - 1)u + 1 - a)b < a q^{r - 1}  +  b u$ and $b 
\neq 0$. Then $a > (q - 2) u + 1$. Again, we fix $\gamma \in \fq^*$ and 
take 
$\beta_1, \ldots, \beta_b \in \fqr$ to be distinct elements such 
that $Tr(\beta_j) = \gamma$ for all $j = 1, \ldots, b$. Let $\alpha_1, \ldots, 
\alpha_a \in \fqr$ be distinct elements such that 
for all $i = 1, \ldots, a$ we have $(\alpha_i, \beta_i) \in \cX_u$ for 
some $\beta_i \in \fqr$ and 
for exactly $a - (q - 2) u - 1$ elements 
$\alpha_i$ we have $\alpha_i^u = \gamma$ (note that since $a < (q-1)u + 1$ we 
get $a - (q - 2) u - 1 < u$). Let $h(x,y) = \prod_{i = 1}^a (x - \alpha_i) 
\cdot \prod_{j = 1}^b (y - \beta_j)$, for each 
$\alpha \in \{\alpha_1, \ldots, \alpha_a\}$ we have $q^{r - 1}$ elements 
$(\alpha, \beta) \in \cX_u$, which are also zeros of $h$. For each $\beta \in 
\{\beta_1, \ldots, \beta_b\}$ we have $u - (a - (q - 2) u - 1) = (q - 1)u + 1 - 
a$ elements of the form  $(\alpha, \beta) \in \cX_u$ which are zeros of $h$ and 
have not been counted yet. Thus, the total number of zeros of $h$ in $\cX_u$ is
$\mid V(I_{\cX_u}  + (h))\mid = a q^{r - 1} + b ((q - 1)u + 1 - 
a)$.

If $b = 0$, then $a q^{r-1} + ((q - 1)u + 1 - a)b = a q^{r - 1}  +  b u$. Taking $\alpha_1, \ldots, \alpha_a \in \fqr$ such 
that for all $i = 1, \ldots, a$, we have $(\alpha_i, \beta_i) \in \cX_u$ for 
some $\beta_i \in \fqr$, the polynomial $t(x,y) = \prod_{i = 1}^a (x 
- \alpha_i)$ is such that $\mid V(I_{\cX_u}  + (t))\mid = a q^{r - 1}$.

Thus, we have proved that for every monomial $x^a y^b \in \cM$, there exists a polynomial $f$ in the $\fqr$-vector space generated by the monomials in $\cM$ having $x^a y^b$ as its leading monomial, and such 
that $\mid V(I_{\cX_u}  + (f))\mid$ attains its greatest possible value, namely
$\min\left\{a q^{r - 1} + (u(q-1) + 1 - a) b , \,  a q^{r - 1} + b u\right\}$. This completes the proof.
\end{proof}

\begin{example}\rm\label{23.06.32}
Take $q=2$, $r=4$, and $u=3$. Figure~\ref{22.09.17}~(a) shows the points of the extended norm-trace curve $\cX_u$. Let $\textcolor{red}{\mathcal{M}}$ be the set of monomials in $\fqr[x,y]$ whose exponents are the points in Figure~\ref{22.09.17}~(b). Note that $\textcolor{red}{\mathcal{M}}$ is closed under divisibility. By Theorem~\ref{22.06.13}, ${\rm ev}(\textcolor{red}{\mathcal{M}})$ is a $[ 32, 12, 12 ]$ decreasing norm-trace code over $\F_{16}$.
\end{example}
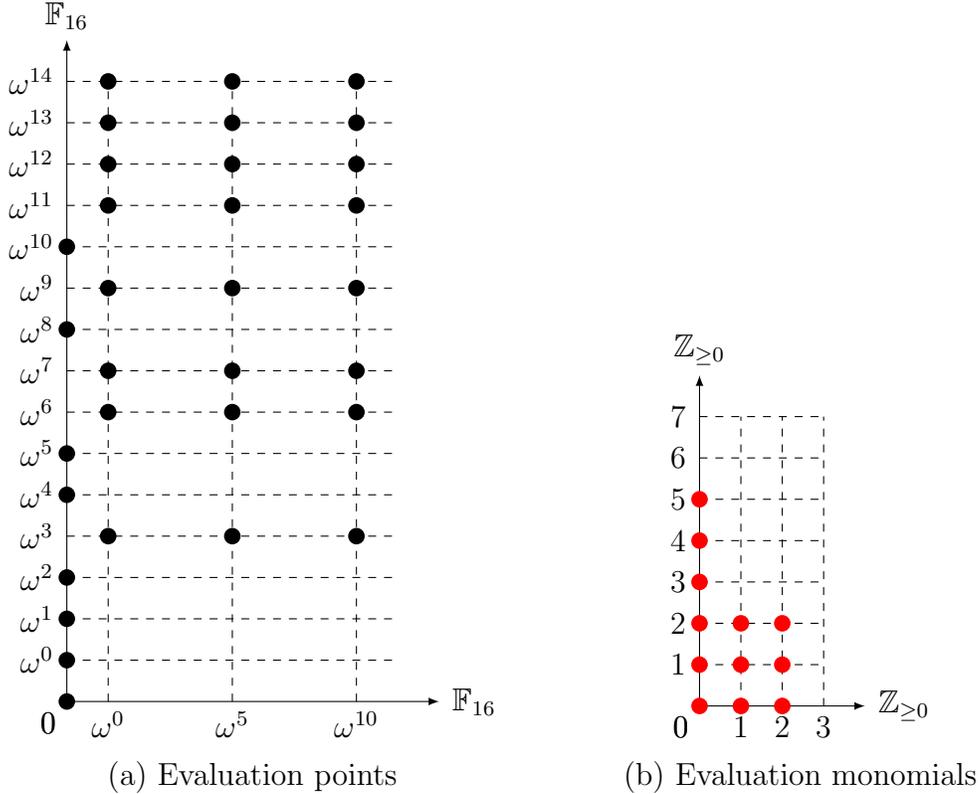
\begin{figure}[h]
\vskip 0cm
\noindent
\begin{minipage}[t]{0.45\textwidth}
\begin{center}
\begin{tikzpicture}[scale=0.55]
\def\xinitial{0}
\def\xfinal{7}
\def\yinitial{0}
\def\yfinal{14}

\draw [-latex] (\xinitial,\yinitial)node[below left]{0} -- (\xinitial,\yfinal+2)node[above]{$\F_{16}$};
\draw [dashed] (0+1,\yinitial)node[below]{$\omega^{0}$} -- (0+1,\yfinal+1)node[right] {};
\draw [dashed] (5-1,\yinitial)node[below]{$\omega^{5}$} -- (5-1,\yfinal+1)node[right] {};
\draw [dashed] (10-3,\yinitial)node[below]{$\omega^{10}$} -- (10-3,\yfinal+1)node[right] {};

\draw [-latex] (\xinitial,\yinitial)node[below left]{0} -- (\xfinal +2,\yinitial)node[right]{$\F_{16}$};
\foreach \i in {\yinitial,...,\yfinal}
{
\draw [dashed] (\xinitial,\i+1)node[left]{$\omega^{\i}$} -- (\xfinal+1,\i+1)node[right] {};
}

\foreach \j in {-1,0,1,2,4,5,8,10}
{
\fill [color=black](0,\j+1) {circle(.2cm)};
}

\foreach \i in {0}
{\foreach \j in {3,6,7,9,11,12,13,14}
{\fill [color=black](\i+1,\j+1) {circle(.2cm)};}}

\foreach \i in {5}
{\foreach \j in {3,6,7,9,11,12,13,14}
{\fill [color=black](\i-1,\j+1) {circle(.2cm)};}}

\foreach \i in {10}
{\foreach \j in {3,6,7,9,11,12,13,14}
{\fill [color=black](\i-3,\j+1) {circle(.2cm)};}}

\end{tikzpicture}
\vskip 0cm
(a) Evaluation points
\end{center}
\end{minipage}
\begin{minipage}[t]{0.45\textwidth}
\begin{center}
\begin{tikzpicture}[scale=0.55]

\def\xinitial{1}
\def\xfinal{3}
\def\yinitial{1}
\def\yfinal{7}

\draw [-latex] (\xinitial,\yinitial)node[below left]{0} -- (\xinitial,\yfinal+2)node[above]{$\Z_{\geq 0}$};
\foreach \i in {\xinitial,...,\xfinal}
{
\draw [dashed] (\i+1,\yinitial)node[below]{$\i$} -- (\i+1,\yfinal+1)node[right] {};
}

\draw [-latex] (\xinitial,\yinitial)node[below left]{0} -- (\xfinal +2,\yinitial)node[right]{$\Z_{\geq 0}$};
\foreach \i in {\yinitial,...,\yfinal}
{
\draw [dashed] (\xinitial,\i+1)node[left]{$\i$} -- (\xfinal+1,\i+1)node[right] {};
}

\foreach \i in {0}
{\foreach \j in {0,...,5}
{\fill [color=red](\i+1,\j+1) {circle(.2cm)};}}

\foreach \i in {1,...,2}
{\foreach \j in {0,...,2}
{\fill [color=red](\i+1,\j+1) {circle(.2cm)};}}

\end{tikzpicture}
\vskip 0cm
(b) Evaluation monomials
\end{center}
\end{minipage}
\caption{Take $q=2$, $r=4$, and $u=3$. (a) Shows the points of the norm-trace curve $\cX_u: x^3=y^8+y^4+y^2+y$. Let $\textcolor{red}{\mathcal{M}}$ be the set of monomials whose exponents are the points in (b). The evaluation code ${\rm ev}(\textcolor{red}{\mathcal{M}})$ is an $[ 32, 12, 12 ]$ decreasing norm-trace code over $\F_{16}$.}
\label{22.09.17}
\vskip 0.cm
\end{figure}

\subsection{Relationship with one-point algebraic geometric codes}\label{23.06.30}
The family of decreasing norm-trace codes contains, as a particular case, the 
family of one-point geometric codes over the norm-trace. Indeed, define 
$\cL_s:=\left\{x^iy^j \in \Delta(I_{\X_u}) \mid iq^{r-1} 
+j\frac{q^r-1}{q-1} \leq s \right\}$. It is 
straightforward to check that $\cL_s$ is closed under divisibility, and that the 
one-point geometric Goppa codes over the norm-trace are 
obtained, as detailed below, through the evaluation, at the points of $\X_u$, 
of the polynomials in the space generated by $\cL_s$
 (see, e.g. \cite{geil,HLP}). We note that the new constructions of 
decreasing norm-trace codes is more general. Example~\ref{23.06.32} is not a 
one-point code. The nearest one-point codes are $C(D,G)$ with $G=\{18, 19, 
20\}P_{\infty}$, which are $[32,\{12,13,14\}, \geq \{14,13,12\}]$ codes.

The extended norm-trace codes introduced and studied in \cite{bras-amor} and \cite{Heera-Pin} are also particular instances of decreasing norm-trace codes.

\begin{remark}
Note that Theorem~\ref{22.06.13} allows us to recover the exact minimum distances of one-point codes on $\mathcal{X}_u$ by choosing specific sets $\mathcal{M} \subseteq \Delta\left(x^{(q-1)u+1}, y^{q^{r-1}}\right)$. For the norm-trace curve, this approach had already 
appeared in \cite{geil}, where such codes are denoted by $E(s)$. Moreover, the improved codes $\tilde{E}(s)$, which appear in \cite{geil}, are also decreasing norm-trace codes, and our results recover those of \cite{geil} with respect to their parameters.
\end{remark}

One may notice that $k+d=$
$$
n+1-\left( max\left( \{ \min\left(a q^{r - 1} + (u(q-1) + 1 - a) b , \,  a q^{r - 1} + b u\right)
\mid x^a y^b \in \cM \}\right) - \mid \mathcal{M} \mid +1\right),$$
meaning decreasing norm-trace codes have a gap of 
$$
 \max\left( \{ \min\left(a q^{r - 1} + 
(u(q-1) + 1 - a) b , \,  a q^{r - 1} + b u\right) \mid x^a y^b \in \cM  
\}\right) - \mid \mathcal{M} \mid +1
$$
to the Singleton Bound $k+d \leq n+1$.

Let $\X$ be the nonsingular curve defined over $\mathbb{F}_{q^r}$ which has 
\[x^{\frac{q^r - 1}{q - 1}} = y^{q^{r - 1}} + y^{q^{r - 2}} + \cdots + y\]
as an affine plane model. Then $\X$ 
has only one point at infinity, say $P_\infty$, which is a rational point, and $\X$ 
has exactly $q^{2r - 1}$ other rational  
points.  The pole divisor of $x$ is 
$q^{r - 1} P_\infty$ and the pole divisor of $y$ is $ \frac{q^r - 1}{q - 1}
P_\infty$. Let $D$ be the divisor which is the sum of the $q^{2r - 1}$ affine 
rational points of $\X$. All AG-codes of the form $C_{\mathcal{L}}(D, s 
P_\infty)$ coincide with codes  $\textrm{ev}(L_s)$, where $L_s$ is the 
$\mathbb{F}_{q^r}$-vector space generated by 
\[ \cL_s = 
\left\{ x^a y^b \, \mid \, 
0 \leq a \leq q^r - 1, 0 \leq b \leq q^{r - 1} - 1,  a q^{r - 1}  + b \frac{q^r 
- 1}{q - 1}  \leq s \right\}
\]
(see \cite[Remark 1]{geil} and the references therein). The dimension of $L_s$ 
is $| A_s |$, and if, for a given $s$, 
there exists a monomial $x^a y^b \in A_s$ such that
$$\min\left(a q^{r - 1} + (q^r - a) b , \,  a q^{r - 1} + b \frac{q^r - 1}{q - 1}\right) = s,$$
then the minimum distance of $L_s$ is equal to $q^{2r - 
1} - \sigma(s),$ where
\[
\sigma(s) = \max \left\{ \min \left(a q^{r - 1} + (q^r - a) b , \,  a q^{r - 1} + b \frac{q^r - 1}{q - 1}\right)
\, \mid \, x^a y^b \in A_s \right\}
\]
(see \cite[Thm.\ 1 and Thm.\ 2]{geil}). We present now some numerical examples 
comparing codes constructed in this paper with the algebraic geometry codes 
described above.

\begin{example}
We start by taking $q = 3$ and $r = 2$.

Choosing $s = 23$, we 
get that 
\[
A_{23}  = \{ x^a y^b \, \mid \, 
0 \leq a \leq 8, 0 \leq b \leq 2,  3  a + 4 b  \leq 23 \}.
\]
Thus $x^5 y^2 \in A_{23}$ and 
$ \min(5 \cdot 3 + 
(9 - 5) \cdot 2 , \,  5 \cdot 3 + 2 \cdot 4) = 23,
$
so the minimum distance is equal to $27 - \sigma(23) = 27 - 23 = 4$.
It is easy to check that
\[A_{23} = \{ x^a y^b \, \mid \, 
0 \leq a \leq 5, 0 \leq b \leq 2\} \cup  \{ x^6, x^6 y, x^7\},\] so
 $\dim (L_{23}) = 21$.
Now we take $\mathcal{M}_{23} := A_{23} \cup \{ x^7 y\}$. Then 
$\mathcal{M}_{23}$ is closed under divisibility, and from the formulas of 
Theorem~\ref{22.06.13}, we get that $\textrm{ev}(\mathcal{M}_{23})$ has dimension 22 and 
minimum distance equal to 4. Both codes have the same length, so  
$\textrm{ev}(\mathcal{M}_{23})$ is better. 

We also obtain a better code choosing $s = 21;$ in this case
$$A_{21} = \{ x^a y^b \, \mid \, 
0 \leq a \leq 4, 0 \leq b \leq 2\} \cup  \{ x^5, x^5 y, x^6, x^7\},$$ so
$\dim(L_{21}) = 19$ and the minimum distance is equal to $27 - 21 = 6$ (here we 
use that $x^7 \in A_{21}$). Taking 
$\mathcal{M}_{21} := A_{21} \cup \{ x^6 y\}$, we get from Theorem~\ref{22.06.13} that 
$\dim(\textrm{ev}(\mathcal{M}_{21})) = 20$ and has minimum distance equal to 6.

Both $\mathcal{M}_{23}$ and $\mathcal{M}_{21}$ have the best minimum distance 
among known codes defined over $\mathbb{F}_{3^2}$ with length $27$ and, 
respectively, dimensions 22 and 20, according to the tables in 
\cite{code_tables}.
\end{example}

By taking higher values for the field over which the codes are defined we may 
obtain more striking differences. 

\begin{example} Take $q = 3$ and $r = 4$, 
so that $q^r = 81$, $q^{r - 1} = 27$ and $(q^r - 1)/(q - 1) = 40$. Set $s = 
57\cdot 27 = 1539$. Unlike the above example, which was computed by hand, in this 
one we used Magma (\cite{magma}) to find out that $| A_{1539} | = 1033$. Thus, 
$\dim(L_{1539}) = 1033$ and using the formulas above (and the fact that $x^{57} 
\in A_{1539}$), we get that the minimum 
distance of this code is $q^{2r - 1} - 1539 = 2187 - 1539 = 648$. Now we take 
\begin{equation*}
\begin{split}
\mathcal{M}_{1539} := A_{1539} \cup \{x^{44} y^9, & x^{45} y^9, x^{46} y^8, 
x^{47} y^7, x^{48} y^7, x^{49} y^6, x^{50} y^5, \\ &x^{50} y^6, x^{51} y^5,
x^{52} y^4, x^{53} y^3, x^{54} y^3, x^{55} y^2, x^{56} y \}.
\end{split}
\end{equation*}
One may check that $\mathcal{M}_{1539}$ is closed under divisibility, and  
using the formulas of Theorem~\ref{22.06.13}, we get that 
$\textrm{ev}(\mathcal{M}_{1539})$ has dimension equal to $1047$ and minimum 
distance equal to $648$.
\end{example}

We also get good codes when $u$ is a proper divisor of $(q^r - 1)/(q - 1)$ as the next example demonstrates. 

\begin{example}
Take $q = 3$ and $r = 2$, but this time we choose $u = 2$, a proper 
divisor of $(q^r - 1)/(q - 1) = 4$. We construct codes of length $((q - 1)u + 1)q^{r - 1} = 15$ using the monomials from the first column of Table~\ref{table:1}.  The parameters of these codes are given in the second and third column of Table~\ref{table:1}. According to the tables in \cite{code_tables}, these are the best minimum 
distances for the corresponding dimensions, considering codes of length 15 defined over $\mathbb{F}_{3^2}$.
\end{example}

\begin{table} 
\centering
{\small \renewcommand{\arraystretch}{1.2}
\noindent
\begin{tabular}{|l|c|c|}
\hline
\hspace{20ex}$\mathcal{M}$ & $\dim(\textrm{ev}(\mathcal{M}))$ & 
min.\ dist.\ of 
$\textrm{ev}(\mathcal{M})$  \\  \hline
$\{1, y \}$ & 2 & 13  \\ \hline
$\{1, y, x \}$  & 3 & 12  \\ \hline
$\{1, y, y^2, x \}$ & 4 & 11  \\ \hline
$\{1, y, y^2, x , x y\}$ & 5 & 10  \\ \hline
$\{1, y, y^2, x , x y, x^2\}$ & 6 & 9  \\ \hline
$\{1, y, y^2, x , x y, x y^2, x^2 \}$ & 7 & 8  \\ \hline
$\{1, y, y^2, x , x y, x y^2, x^2, x^2 y\}$ & 8 & 7  \\ \hline
$\{1, y, y^2, x , x y, x y^2, x^2, x^2 y, x^3\}$ & 9 & 6  \\ \hline
$\{1, y, y^2, x , x y, x y^2, x^2, x^2 y, x^2 y^2, x^3 \}$ & 10 & 5  \\ \hline
$\{1, y, y^2, x , x y, x y^2, x^2, x^2 y, x^2 y^2, x^3,x^3 y \}$ & 11 & 4  \\ 
\hline
$\{1, y, y^2, x , x y, x y^2, x^2, x^2 y, x^2 y^2, x^3,x^3 y, x^4 \}$ & 12 & 3  
\\ \hline
\end{tabular}
\caption{Code parameters.}
\label{table:1}
}
\vspace{2ex}
\end{table}

\noindent

\section{Duals of decreasing norm-trace codes}\label{22.09.14}
This section proves that the dual of a decreasing norm-trace code is equivalent to a decreasing norm-trace code. In addition, we describe the dual code in terms of the monomial set and the coefficients of the indicator functions. We then give conditions to find families of self-dual and self-orthogonal codes.
 
Recall that the two linear codes $C_1$ and $C_2$ in $\F_{q^r}^n$
are \textit{equivalent} if there is ${\bm 
\beta}=(\beta_1,\ldots,\beta_n) \in \F_{q^r}^n$ such that $\beta_i\neq 0$ for all $i$ and 
$C_2 = {\bm \beta} \cdot C_1 := \{\beta\cdot c\mid c\in C_1\}$, 
where ${\bm \beta} \cdot c: =(\beta_1c_1,\ldots,\beta_nc_n)$ for 
$c=(c_1,\ldots,c_n)\in C_1$. Some authors call such codes monomially equivalent, but we will simply say equivalent as no other type of equivalence is considered in this paper.

We come to one of the main results of this work, which computes the dual of a decreasing norm-trace code. Recall that the two linear codes $C_1$ and $C_2$ in $\F_{q^r}^n$
are \textit{equivalent} if there is ${\bm 
\beta}=(\beta_1,\ldots,\beta_n) \in \F_{q^r}^n$ such that $\beta_i\neq 0$ for all $i$ and 
$C_2 = {\bm \beta} \cdot C_1 := \{\beta\cdot c\mid c\in C_1\}$, 
where ${\bm \beta} \cdot c: =(\beta_1c_1,\ldots,\beta_nc_n)$ for 
$c=(c_1,\ldots,c_n)\in C_1$. 

\begin{theorem}\label{22.03.15}
Assume $\cX_u = \left\{P_1,\ldots,P_n \right\} $ and let ${\rm ev}(\mathcal{M})$ be a decreasing norm-trace code.
The dual code ${\rm ev}\left(\mathcal{M}\right)^\perp$ is equivalent to the code ${\rm ev}\left(\mathcal{M}^\complement\right)$ where
\[\mathcal{M}^\complement :=
\left\{\frac{x^{(q-1)u} y^{q^{r-1}-1}}{x^iy^j} : x^iy^j \in 
\Delta\left(x^{(q-1)u+1}, y^{q^{r-1}}\right) \setminus \mathcal{M} \right\} \] 
denotes the complement of $\mathcal{M}$. 
More precisely,
\[{\rm ev}\left(\mathcal{M}\right)^\perp={\bm \beta}\cdot {\rm ev}\left(\mathcal{M}^\complement\right),\]
where $\beta_i:= \begin{cases} u^{-1} & \textnormal{if the }x\textnormal{-coordinate of } P_i \textnormal{ is nonzero} \\ 1 & \textnormal{otherwise.} \end{cases}$
\end{theorem}
\begin{proof}
From Theorem \ref{22.06.05}, we get that $y^{q^{r - 1}-1}x^{(q - 1)u}$ is the 
leading 
monomial of the standard indicator function for all points in $\cX_u$.
We also note that 
\[ \bigm| \mathcal{M}\bigm| + \bigm| \mathcal{M}^\complement \bigm| = \bigg| 
\Delta\left(x^{(q-1)u+1}, y^{q^{r-1}}\right) \bigg| = \mid \cX_u \mid.\]
Thus, according to \cite[Theorem 5.4]{dual2} (\cite[Theorem 5.4]{dual2} is an updated version of \cite[Theorem 5.4]{dual}), to prove the Theorem it suffices 
to prove that, given $x^{a} y^{b} \in \mathcal{M}$ and
$x^{c} y^{d}\in \mathcal{M}^\complement$,  the coefficient of the monomial 
$x^{(q-1)u} y^{q^{r - 1} - 1}$ in the unique 
$\fqr$-linear combination  of monomials in $\Delta_{\prec}(I_{\cX_u})$ 
which has the same class, in $\fqr[x, y]/I_{\X_u}$,  as $x^{a + c} y^{b + d}$, 
is equal to zero.

Before we prove that, we claim that if $a + c \geq (q - 1)u$, then $b + d < q^{r 
- 1} -1$. Indeed, 
assume that $b + d \geq q^{r 
- 1} -1$. We set $c' = (q-1)u - c$ and $d' = q^{r - 1} - 1 - d$, and from the 
definition of  $\mathcal{M}^\complement$, we get that $x^{c'} 
y^{d'} \notin \mathcal{M}$. From $d \geq q^{r - 1} - 1  - b$, we get $d' \leq b$.
Since $\mathcal{M}$ is closed under divisibility, we must have $c' > a$, so 
$(q-1)u - c > a$ and $a + c < (q-1)u$, which proves the claim.

This shows that either $a + c \geq (q - 1)u$ and $b + d < q^{r - 1} -1$, or 
$b + d \geq q^{r - 1} -1$ and $a + c < (q - 1)u$. 

Assume that $a + c > (q - 1)u$. Since $a + c \leq 2 (q - 1)u$, we write $a + c = 
((q - 1) u + 1) + r$ with $0 \leq r \leq (q - 1)u -1$. Then $x^{a + c} + 
I_{\cX_u} = x^{(q - 1) u + 1} x^r +  I_{\cX_u} = x^{r + 1} +  I_{\cX_u}$ and we 
get, in this case, that $x^{a + c} y^{b + d} + I_{\cX_u} = x^{r + 1} y^{b + d} 
+ I_{\cX_u}$ with $x^{r + 1} y^{b + d} \in \Delta_{\prec}(I_{\cX_u}) \setminus 
\{x^{(q-1)u} y^{q^{r - 1} - 1}\}$, since $b + d < q^{r - 1} -1$.

Now assume that $b + d > q^{r - 1} -1$ (and then $a + c < (q - 1)u$). 
Note 
that $b + d  \leq 2 q^{r - 1} - 2$, so we 
write $b + d = 2 q^{r - 1} - 2 - e$, with $0 \leq e \leq q^{r - 1} - 2$.
We have
\begin{equation*}
\begin{split}
	y^{2 q^{r - 1} - 2 - e}  = &y^{q^{r - 1} - 2 - e}(y^{q^{r - 1}} + 
		y^{q^{r - 2}} + \cdots + y^q + y - x^u) \\ &- y^{q^{r - 1} - 2 - 
		e}(y^{q^{r - 2}} + \cdots + y^q + y - x^u)
\end{split}
\end{equation*}
so that 
\[
x^{a + c} y^{b + d} + I_{\cX_u} = - x^{a + c}y^{q^{r - 1} - 2 - 
		e}(y^{q^{r - 2}} + \cdots + y^q + y - x^u)  + I_{\cX_u}.
\]
If $q^{r - 1} + q^{r - 2} - 2 - e \leq q^{r - 1} - 1$ then all the monomials in 
the 
polynomial $x^{a + c}y^{q^{r - 1} - 2 - e}(y^{q^{r - 2}} + \cdots + y^q + y - 
x^u)$ are in $\Delta_{\prec}(I_{\cX_u}) \setminus 
\{x^{(q-1)u} y^{q^{r - 1} - 1}\}$, except possibly for $x^{a + b + u}$, which 
may be reduced, as above, to a power of $x$ in $\Delta_{\prec}(I_{\cX_u})$. 

The last case to consider is when $q^{r - 1} + q^{r-2} - 2 - e \geq q^{r - 1}$. 
In this case, we have to find proper representatives for some monomials in
$x^{a + c}y^{q^{r - 1} - 2 - e}(y^{q^{r - 2}} + \cdots + y^q + y - x^u)$. We 
write $q^{r - 1} + q^{r-2} - 2 - e = 2 q^{r - 1} - 2 - e'$, where $e' = q^{r - 
1} - q^{r - 2} + e$, and we have $q^{r - 1}  - q^{r - 2} \leq e' \leq q^{r - 1} 
- 2$. From 
\begin{equation*}
\begin{split}
	y^{2 q^{r - 1} - 2 - e'}  = &y^{q^{r - 1} - 2 - e'}(y^{q^{r - 1}} + 
		y^{q^{r - 2}} + \cdots + y^q + y - x^u) \\ &- y^{q^{r - 1} - 2 - 
		e'}(y^{q^{r - 2}} + \cdots + y^q + y - x^u),
\end{split}
\end{equation*}
we get that 
$x^{a + c} y^{2 q^{r - 1} - 2 - e'} + I_{\cX_u} = - x^{a + c}y^{q^{r - 1} - 2 
- e'}(y^{q^{r - 2}} + \cdots + y^q + y - x^u)  + I_{\cX_u}$. Since 
$q^{r - 1} + q^{r - 2} - 2 - e' \leq 2 q^{r - 2} - 2 < q^{r - 1} - 1$, we get 
that the monomials in the polynomial $x^{a + c}y^{q^{r - 1} - 2 
- e'}(y^{q^{r - 2}} + \cdots + y^q + y - x^u)$ are in 
$\Delta_{\prec}(I_{\cX_u}) \setminus 
\{x^{(q-1)u} y^{q^{r - 1} - 1}\}$, except possibly for $x^{a + b + u}$ which 
may be represented by a 
power of $x$ in $\Delta_{\prec}(I_{\cX_u})$. This takes care of the monomial
$x^{a + c}y^{q^{r - 1} + q^{r - 2} - 2 - e}$ in  
$x^{a + c}y^{q^{r - 1} - 2 - e}(y^{q^{r - 2}} + \cdots + y^q + y - x^u)$. In 
the same way we prove that if, for some integer $s \geq 3$, we get $q^{r - 1} + 
q^{r-s} - 2 - e \geq q^{r - 1}$, then 
$y^{q^{r - 1} + q^{r - s} - 2 - e} + I_{\cX_u} = - y^{q^{r - 1} - 2 - 
e''}(y^{q^{r - 2}} + \cdots + y^q + y - x^u) + I_{\cX_u}$, where 
$q^{r - 1}  - q^{r - s} \leq e'' \leq q^{r - 1} - 2$, so 
$q^{r - 1} + q^{r - s} - 2 - e'' \leq  q^{r - 2} + q^{r - s} - 2 < q^{r - 1} - 
1$.	This proves that, also in the case where $q^{r - 1} + q^{r-2} - 2 - e \geq 
q^{r - 1}$,  the unique 
$\fqr$-linear combination  of monomials in $\Delta_{\prec}(I_{\cX_u})$, 
which has the same class as $x^{a + c}y^{q^{r - 1} - 2 - e}(y^{q^{r - 2}} + \cdots + y^q + y - x^u)$ in $\fqr[x, y]/I_{\X_u}$,
has zero as the coefficient of the monomial $x^{(q-1)u} y^{q^{r - 1} - 1}$, 
which completes the proof of the Theorem.
\end{proof}

\begin{example}\label{22.03.12} \rm
Take $q=3$ and $r=2$.Figure~\ref{22.03.13}~(a) shows the points of the norm-trace curve
$\cX: x^4=y^3+y$.  Let $\textcolor{red}{\mathcal{M}}$ be the set of 
monomials in $\Delta\left(x^{9}, y^{3}\right)$ of degree at most 4. The exponents of these monomials 
are the points in Figure~\ref{22.03.14}~(a). The complement of 
$\textcolor{red}{\mathcal{M}}$ on $\cX$ is the set of monomials 
$\textcolor{blue}{\mathcal{M}^\complement}$, whose exponents are the points in 
Figure~\ref{22.03.14}~(b). By Theorem~\ref{22.03.15}, the dual code
${\rm ev}\left(\textcolor{red}{\mathcal{M}}\right)^\perp$ is equivalent to 
the code ${\rm ev}\left(\textcolor{blue}{\mathcal{M}^\complement}\right)$.
\end{example}
\begin{figure}[h]
\vskip 0cm
\noindent
\begin{minipage}[t]{0.45\textwidth}
\begin{center}
\begin{tikzpicture}[scale=0.55]
\def\xinitial{1}
\def\xfinal{8}
\def\yinitial{1}
\def\yfinal{2}

\draw [-latex] (\xinitial,\yinitial)node[below left]{0} -- (\xinitial,\yfinal+2)node[above]{$\Z_{\geq 0}$};
\foreach \i in {\xinitial,...,\xfinal}
{
\draw [dashed] (\i+1,\yinitial)node[below]{$\i$} -- (\i+1,\yfinal+1)node[right] {};
}

\draw [-latex] (\xinitial,\yinitial)node[below left]{0} -- (\xfinal +2,\yinitial)node[right]{$\Z_{\geq 0}$};
\foreach \i in {\yinitial,...,\yfinal}
{
\draw [dashed] (\xinitial,\i+1)node[left]{$\i$} -- (\xfinal+1,\i+1)node[right] {};
}

\foreach \i in {0,1,2}
{\foreach \j in {0,...,2}
{\fill [color=red](\i+1,\j+1) {circle(.2cm)};}}

\foreach \i in {3}
{\foreach \j in {0,1}
{\fill [color=red](\i+1,\j+1) {circle(.2cm)};}}

\foreach \i in {4}
{\foreach \j in {0}
{\fill [color=red](\i+1,\j+1) {circle(.2cm)};}}

\end{tikzpicture}
\vskip 0cm
(a) Monomials of degree at most 4.
\end{center}
\end{minipage}
\begin{minipage}[t]{0.45\textwidth}
\begin{center}
\begin{tikzpicture}[scale=0.55]

\draw [-latex] (8,2)node[above right]{0} -- (-1,2)node[left] {$\Z_{\geq 0}$};
\draw [dashed] (0,1)node[left]{} -- (8,1)node[right] {1};
\draw [dashed] (0,0)node[left]{} -- (8,0)node[right] {2};

\draw [dashed] (0,0)node[below]{} -- (0,2)node[above] {8};
\draw [dashed] (1,0)node[below]{} -- (1,2)node[above] {7};
\draw [dashed] (2,0)node[below]{} -- (2,2)node[above] {6};
\draw [dashed] (3,0)node[below]{} -- (3,2)node[above] {5};
\draw [dashed] (4,0)node[below]{} -- (4,2)node[above] {4};
\draw [dashed] (5,0)node[below]{} -- (5,2)node[above] {3};
\draw [dashed] (6,0)node[below]{} -- (6,2)node[above] {2};
\draw [dashed] (7,0)node[below]{} -- (7,2)node[above] {1};
\draw [-latex] (8,2) -- (8,-1)node[below] {$\Z_{\geq 0}$};

\fill [color=blue](8,2) {circle(.2cm)};
\fill [color=blue](8,1) {circle(.2cm)};
\fill [color=blue](8,0) {circle(.2cm)};
\fill [color=blue](7,2) {circle(.2cm)};
\fill [color=blue](7,1) {circle(.2cm)};
\fill [color=blue](7,0) {circle(.2cm)};
\fill [color=blue](6,2) {circle(.2cm)};
\fill [color=blue](6,1) {circle(.2cm)};
\fill [color=blue](6,0) {circle(.2cm)};
\fill [color=blue](5,2) {circle(.2cm)};
\fill [color=blue](5,1) {circle(.2cm)};
\fill [color=blue](5,0) {circle(.2cm)};
\fill [color=blue](4,2) {circle(.2cm)};
\fill [color=blue](4,1) {circle(.2cm)};
\fill [color=blue](3,2) {circle(.2cm)};

\end{tikzpicture}
\vskip 0.cm
(b) Complement of (a) on $\cX$.
\end{center}
\end{minipage}
\caption{(a) shows the exponents of the set of monomials $\textcolor{red}{\mathcal{M}}$ 
in $\Delta\left(x^{9}, y^{3}\right)$ of degree at most 4. (b) shows the exponents of 
$\textcolor{blue}{\mathcal{M}^\complement}$, the complement of 
$\textcolor{red}{\mathcal{M}}$ on $\cX$. By Theorem~\ref{22.03.15}, the dual code
${\rm ev}\left(\textcolor{red}{\mathcal{M}}\right)^\perp$ is equivalent to 
the code ${\rm ev}\left(\textcolor{blue}{\mathcal{M}^\complement}\right)$.}
\label{22.03.14}
\end{figure}
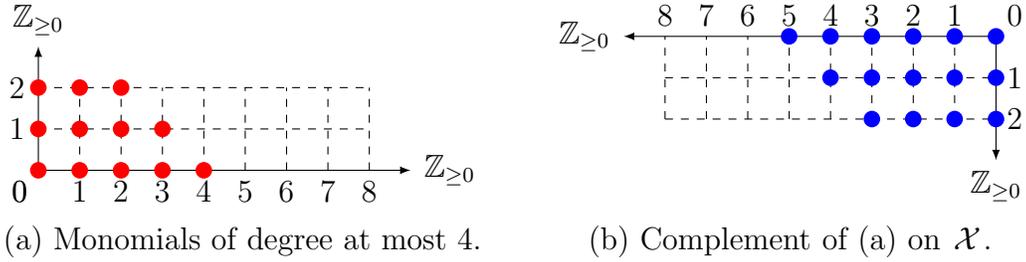

Recall that the hull of a code $C$ is  $\hull(C):=C\cap C^{\perp}$.  The code $C$ is {\it self-dual} if $C= C^{\perp}$ and {\it self-orthogonal} if $C \subseteq C^{\perp}$. Theorem~\ref{22.03.15} gives a powerful tool for designing self-dual and self-orthogonal codes.
\begin{theorem}\label{22.06.12}
Assume $\cX_u = \left\{P_1,\ldots,P_n \right\} $ and let ${\rm ev}(\mathcal{M})$ be a decreasing norm-trace code.
If the equation $x^2=u$ has a solution $\alpha$ in $\fqr$, then
\[\hull({\bm \lambda} \cdot {\rm ev}(\cM))= {\bm \lambda} \cdot {\rm ev}\left(\cM \cap \cM^\complement \right),\]
where $\lambda_i := \begin{cases} \alpha^{-1} &\textnormal{if the }x\textnormal{-coordinate of }P_i \textnormal{ is nonzero} \\ 1 &\textnormal{ otherwise.} \end{cases}$
\end{theorem}
\begin{proof}
Denote by ${\bm \lambda}^{-1}$ the vector whose entries are $\lambda_i^{-1}$. By Theorem~\ref{22.03.15}, we have that ${\bm \lambda}{\bm \lambda}={\bm \beta}$, so ${\bm \lambda}={\bm \lambda}^{-1}{\bm \beta}$, where the product between vectors is pointwise. Thus, $\left({\bm \lambda} \cdot {\rm ev}(\cM) \right)^\perp = {\bm \lambda^{-1}} \cdot {\rm ev}\left(\cM \right)^\perp ={\bm \lambda^{-1}} {\bm \beta}\cdot {\rm ev}\left(\mathcal{M}^\complement\right) ={\bm \lambda} \cdot {\rm ev}\left(\mathcal{M}^\complement\right)$.
\end{proof}

\begin{corollary}\label{22.09.13}
Assume $\cX_u = \left\{P_1,\ldots,P_n \right\}$ and the equation $x^2=u$ has a 
solution $\alpha$ in $\fqr$. If $\cM \subseteq \cM^\complement$, then ${\bm 
\lambda} \cdot {\rm ev}(\mathcal{M})$ is a self-orthogonal code, where 
$\lambda_i$ is as in Theorem \ref{22.06.12}. If $\cM = \cM^\complement$, then 
${\bm \lambda} \cdot {\rm ev}(\mathcal{M})$ is a self-dual code.
\end{corollary}
\begin{proof}
If $\cM \subseteq \cM^\complement$, then
$\hull({\bm \lambda} \cdot {\rm ev}(\cM))= {\bm \lambda} \cdot {\rm ev}\left(\cM \cap \cM^\complement \right)=
{\bm \lambda} \cdot {\rm ev}\left(\cM \right)$
by Theorem~\ref{22.06.12}. Thus,
${\bm \lambda} \cdot {\rm ev}\left(\cM \right)=
\hull({\bm \lambda} \cdot {\rm ev}(\cM)) \subseteq \left({\bm \lambda} \cdot {\rm ev}\left(\cM \right)\right)^\perp$.
The case $\cM = \cM^\complement$ is analogous.
\end{proof}
\begin{example}\label{22.03.19} \rm
Take $q=2$, $r=4$, and $u=5$. Figure~\ref{22.03.20}~(a) shows the points of the set $\X_u$. Let $\textcolor{red}{\mathcal{M}}$ be the set of monomials in $\fqr[x,y]$ with degree in $x$ at most $5$ and degree in $y$ at most $4$. The exponents of these monomials are the points in Figure~\ref{22.03.20}~(b). As $5\equiv1$ in $\F_2$, then $\hull({\rm ev}(\textcolor{red}{\cM}))= {\rm ev}\left(\textcolor{red}{\cM} \cap \cM^\complement \right)={\rm ev}(\textcolor{red}{\cM})$ by Theorem~\ref{22.06.12}. Thus, ${\rm ev}(\textcolor{red}{\cM})$ is a self-dual code by Corollary~\ref{22.09.13}.
\end{example}
\begin{figure}[h]
\vskip 0cm
\noindent
\begin{minipage}[t]{0.45\textwidth}
\begin{center}
\begin{tikzpicture}[scale=0.35]
\def\xinitial{0}
\def\xfinal{12}
\def\yinitial{0}
\def\yfinal{14}

\draw [-latex] (\xinitial,\yinitial)node[below left]{0} -- (\xinitial,\yfinal+2)node[above]{$\F_{16}$};
\foreach \i in {0,3,6}
{
\draw [dashed] (\i+1,\yinitial)node[below]{$\omega^{\i}$} -- (\i+1,\yfinal+1)node[right] {};
}
\foreach \i in {10,13}
{
\draw [dashed] (\i,\yinitial)node[below]{$\omega^{\i}$} -- (\i,\yfinal+1)node[right] {};
}

\draw [-latex] (\xinitial,\yinitial)node[below left]{0} -- (\xfinal +2,\yinitial)node[right]{$\F_{16}$};
\foreach \i in {\yinitial,...,\yfinal}
{
\draw [dashed] (\xinitial,\i+1)node[left]{$\omega^{\i}$} -- (\xfinal+1,\i+1)node[right] {};
}

\foreach \j in {-1,0,1,2,4,5,8,10}
{
\fill [color=black](0,\j+1) {circle(.2cm)};
}

\foreach \i in {0,3,6}
{\foreach \j in {3,6,7,9,11,12,13,14}
{\fill [color=black](\i+1,\j+1) {circle(.2cm)};}}

\foreach \i in {10,13}
{\foreach \j in {3,6,7,9,11,12,13,14}
{\fill [color=black](\i,\j+1) {circle(.2cm)};}}

\end{tikzpicture}
\vskip 0cm
(a) Evaluation points
\end{center}
\end{minipage}
\begin{minipage}[t]{0.45\textwidth}
\begin{center}
\begin{tikzpicture}[scale=0.55]

\def\xinitial{1}
\def\xfinal{5}
\def\yinitial{1}
\def\yfinal{7}

\draw [-latex] (\xinitial,\yinitial)node[below left]{0} -- (\xinitial,\yfinal+2)node[above]{$\Z_{\geq 0}$};
\foreach \i in {\xinitial,...,\xfinal}
{
\draw [dashed] (\i+1,\yinitial)node[below]{$\i$} -- (\i+1,\yfinal+1)node[right] {};
}

\draw [-latex] (\xinitial,\yinitial)node[below left]{0} -- (\xfinal +2,\yinitial)node[right]{$\Z_{\geq 0}$};
\foreach \i in {\yinitial,...,\yfinal}
{
\draw [dashed] (\xinitial,\i+1)node[left]{$\i$} -- (\xfinal+1,\i+1)node[right] {};
}

\foreach \i in {0,...,5}
{\foreach \j in {0,...,3}
{\fill [color=red](\i+1,\j+1) {circle(.2cm)};}}

\end{tikzpicture}
\vskip 0cm
(b) Evaluation monomials
\end{center}
\end{minipage}
\caption{(a) shows the points of the curve $\cX: x^u=y^8+y^4+y^2+y$. Let $\textcolor{red}{\mathcal{M}}$ be the set of monomials whose exponents are the points in (b). The evaluation code ${\rm ev}(\textcolor{red}{\mathcal{M}})$ is a self-dual code over $\F_{16}$.}
\label{22.03.20}
\vskip 0.cm
\end{figure}
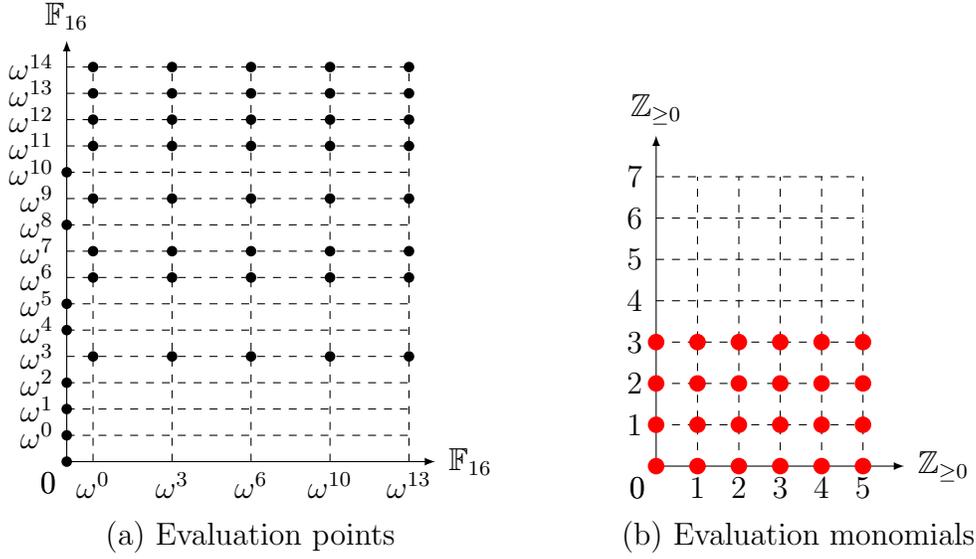

\section{Single Erasure Repair Scheme}\label{erasure}
This section defines a repair scheme that repairs a single erasure for specific 
decreasing norm-trace codes. An element of $\mathbb{F}_{q^r}$ may be thought 
of as a vector in $\mathbb{F}_q^r$. In this theory, the elements of 
$\mathbb{F}_{q^r}$  are called {\em symbols} and the elements of 
$\mathbb{F}_q$ are called {\em subsymbols}. Given a code $C \subset 
\mathbb{F}_{q^r}^n$, a repair scheme is an algorithm that recovers the entry of
any vector of $C$ using the other entries. The {\em bandwidth} $b$ is the 
number of 
subsymbols required by the algorithm to repair the entry. A codeword is defined 
by $n r$ subsymbols, and the fraction $\displaystyle \frac{b}{n r}$  is called 
{\em bandwidth rate}.

Recall that $\Delta\left(x^i, y^j \right)$ denotes the set of monomials which 
are not multiples of either of these two monomials.
Using that $\{ M + I_{\cX_u} \mid M \in \Delta_{\prec}\left(I_{\cX_u}\right) 
\}$ is a 
basis for $\fqr[x,y]/I_{\X_u}$ as an $\fqr$-vector space, we may assume that an 
arbitrary element of ${\rm ev}(\mathcal{M})$
is of the type $\textrm{ev}(f)$, where every monomial that appears in $f$ is in
$\Delta_{\prec}(I_{\cX_u}) = \Delta\left(x^{(q-1)u+1}, y^{q^{r-1}}\right)$. 
Take $n:=((q-1)u+1)q^{r-1}$. Since $\textrm{ev}(f) \in \fqr^{n}$, the element $\textrm{ev}(f)$ depends on $n$ symbols (over $\fqr$) or, equivalently, on $nr$ subsymbols (over $\fq$).
 
\begin{remark}\cite[Definition 2.30 and Theorem 2.40]{Finite_Fields_Book}\label{Dual Bases Remark}
Let $\mathcal B=\{z_1, \dots, z_r\}$ be a basis of $\fqr$ over $\mathbb{F}_q.$ Then there exists a basis  $\{z^\prime_1, \dots, z^\prime_r\}$ of $\fqr$ over $\mathbb{F}_q$, called the dual basis of $\mathcal B$, such that $Tr(z_i z^\prime_j)=\delta_{ij}$ is a delta function and for $\alpha \in \fqr$, \[\alpha = \sum_{i=1}^r Tr(\alpha z_i)z_i^\prime.\]
Thus, determining $\alpha$ is equivalent to finding $Tr(\alpha z_i)$ for $i \in \{ 1, \dots, r \}$.
\end{remark}
\begin{theorem}\label{22.06.08}
Let $\mathcal{M} \subseteq \Delta\left(x^{(q-1)u}, y^{q^{r-1}}\right)$ be a monomial set that is closed under divisibility.
There exists a repair scheme of ${\rm ev}(\mathcal{M})$ for one erasure with bandwidth at most
\[ \mid \cX_u \mid -1 + (u - 1)(r-1).\]
\end{theorem}
\begin{proof}
Take ${\cX_u}= \{P_1,\ldots,P_n \}$ and let $\textrm{ev}(f)=(f(P_1),\ldots,f(P_n))$ be an element of ${\rm ev}(\mathcal{M})$. Assume that the coordinate $f(P^*)$ of $\textrm{ev}(f)$ is erased, where $P^*=(\alpha^*,\beta^*) \in \cX_u$.
We define the following polynomials
\begin{equation*}
p_i(y)=
\frac{Tr(z_i(y - \beta^*))}{(y - \beta^*)}
= {z_i} +
z_i^q(y - \beta^*)^{q-1}+ \cdots +
z_i^{q^{r-1}}(y - \beta^*)^{q^{r-1}-1}
\end{equation*}
for $i\in [r]$. We have  $\{1,y,\ldots, y^{q^{r-1} - 1} \} \subseteq \mathcal{M}^\complement$, as $\mathcal{M} \subseteq \Delta\left(x^{(q-1)u}, y^{q^{r-1}}\right)$. The element ${\bm \beta} \cdot  \left(p_i(P_1),\ldots,p_i(P_n) \right)$ is in ${\rm ev}\left( \mathcal{M} \right)^\perp$ for $i\in [r]$ and ${\bm \beta}$ defined in Theorem~\ref{22.03.15}.  Therefore, we obtain the $r$ equations
\begin{equation}\label{21.06.15}
\beta_{P^*}p_{i}(P^*)f(P^*)= -\sum_{\cX_u \setminus\{P^*\}}
\beta_{P} p_{i}(P)f(P),\quad i\in[r].
\end{equation}
As $p_i(P^*)=z_i,$ applying the trace function to both sides of previous equations and employing the linearity of the trace function, we obtain
\[Tr \left(z_i \beta_{P^*} f(P^*)\right)= -\sum_{ \cX_u \setminus\{P^*\}}
Tr \left( \beta_{P} p_{i}(P)f(P) \right),\quad i\in[r].\]
Define the set $ \Gamma = \{(\alpha,\beta)\in \cX_u : \beta = \beta^* \}.$
We have that $p_i(P)=z_i$ for $ P \in \Gamma$.
For $ P=(\alpha,\beta) \in \cX_u \setminus \Gamma,$ $p_i(P)=\displaystyle \frac{ Tr(z_i(\beta - \beta^*))}{(\beta - \beta^*)}.$ We have that for $i\in[r]$,
\begin{eqnarray*}
\sum_{\cX_u \setminus\{P^*\}} Tr \left( \beta_{P} p_{i}(P)f(P) \right)
&=& \sum_{\Gamma \setminus\{P^*\}} Tr \left( \beta_{P} p_{i}(P)f(P) \right) + 
\sum_{\cX_u \setminus \Gamma} Tr \left( \beta_{P} p_{i}(P)f(P) \right) \\
&=& \sum_{\Gamma \setminus\{P^*\}} Tr \left( \beta_{P} z_i f(P) \right) + 
\sum_{\cX_u \setminus \Gamma} Tr \left( \beta_{P} \frac{Tr(z_i(\beta - \beta^*))}{(\beta - \beta^*)}f(P) \right)\\
&=& \sum_{\Gamma \setminus\{P^*\}} Tr \left( \beta_{P} z_i f(P) \right) + 
\sum_{\cX_u \setminus \Gamma} Tr(z_i(\beta - \beta^*)) Tr \left( \frac{\beta_{P} f(P)}{(\beta - \beta^*)} \right).
\end{eqnarray*}
The element $\beta_{P^*} f(P^*),$ and $f(P^*)$ as a consequence, can be recovered from its $r$ independent traces $Tr (z_i \beta_{P^*} f(P^*))$ by Remark~\ref{Dual Bases Remark}. The traces are obtained by downloading:
\begin{itemize}
\item For each $P\in \Gamma\setminus\{P^*\}$, the $r$ subsymbols $Tr\left(\beta_{P}z_1 f(P)\right), \ldots, Tr\left(\beta_{P}z_r f(P)\right)$.
\item For each $P\in \cX_u\setminus \Gamma$, the subsymbol $\displaystyle Tr\left(\frac{\beta_{P} f(P)}{ (\beta - \beta^*)}\right)$.
\end{itemize}
Hence, the bandwidth is $\displaystyle b = r (\mid\Gamma\mid-1) + \mid\cX_u\setminus \Gamma\mid $ $ \leq
r\left(u-1\right)+ \mid \cX_u \mid - u =\mid \cX_u \mid + (u - 1)(r-1)-1$. 
\end{proof}

Consequently, we obtain the following result for the norm-trace curve.
\begin{corollary}\label{22.06.09}
If $\mathcal{M} \subseteq \Delta\left(x^{q^r}, y^{q^{r-1}-1}\right)$ or
$\mathcal{M} \subseteq \Delta\left(x^{q^r-1}, y^{q^{r-1}}\right)$ is a monomial set that is closed under divisibility, then there exists a repair scheme of the decreasing norm-trace code ${\rm ev}(\mathcal{M})$ for one erasure with bandwidth at most
\[ \mid \cX_u \mid -1 + \left(\frac{q^r-1}{q-1} - 1\right)(r-1).\]
In particular, there exists a repair scheme for the Hermitian decreasing code for one erasure with bandwidth at most
\[ q^3 + q -1.\]
\end{corollary}
\begin{proof}
This is a consequence of Theorem~\ref{22.06.08} for the particular case when $u=\frac{q^r-1}{q-1}$. The Hermitian case is obtained when $r=2$.
\end{proof}

Jin et al. introduced in~\cite{JLX} a repair scheme for single erasures of algebraic geometry codes. In particular,~\cite[Theorem 3.3]{JLX} repairs a single erasure on one-point AG codes defined over the curve $\cX_u$, which can also be considered as monomial decreasing norm-trace codes~\cite{Eduardo_polar}.  Both schemes,~\cite[Theorem 3.3]{JLX} and~Theorem~\ref{22.06.08}, have restrictions and can repair codes with up to a maximum dimension.
One of the main advantages of Theorem~\ref{22.06.08} is the ability to repair single erasures on codes with a higher dimension that use the rational points of the curve $\cX_u$ as evaluation points. Indeed, consider the case where we want to repair an erasure on a monomial decreasing norm-trace code ${\rm ev}(\mathcal{M})$. By Theorem~\ref{22.06.13}, the length of the code ${\rm ev}(\mathcal{M})$ is $n = \mid \cX_u \mid = ((q - 1)u + 1)q^{r-1}$. By the hypothesis of Theorem~\ref{22.06.08}, the maximum dimension where the repair scheme can be applied is when $\mathcal{M} = \Delta\left(x^{(q-1)u}, y^{q^{r-1}}\right)$, where the dimension is
\begin{equation}\label{22.09.11}
k_{ev}:=(q - 1)uq^{r-1}=\mid \cX_u \mid-q^{r-1}.
\end{equation}

Now, consider the case where we want to repair an erasure on a one-point AG 
code over the curve $\cX_u$. The curve $\cX_u$ has genus 
$\mathfrak{g}:=\frac{(u-1)(q^{r-1}-1)}{2}$ (see \cite[Thm. 13]{miura}). In the 
context of~\cite[Theorem 3.3]{JLX}, the maximum dimension of the one-point AG 
code where the repair scheme can be applied is when $m = \mid \cX_u \mid - 
(q-1)(\mathfrak{g}+1)$, which implies that the dimension would be
\begin{equation}\label{22.09.12}
k_{AG}:=m-\mathfrak{g} = \mid \cX_u \mid - (q-1)(\mathfrak{g}+1) - \mathfrak{g} =  \mid \cX_u \mid - q(\mathfrak{g}-1)+1.
\end{equation}
\begin{example}
Taking $u=\frac{q^r-1}{q-1}$, we have that $\mathfrak{g}:=\frac{(u-1)(q^{r-1}-1)}{2}=\frac{\left(\frac{q^r-1}{q-1}-1\right)(q^{r-1}-1)}{2}$. From Equation~\ref{22.09.12}, $k_{AG}=\mid \cX_u \mid - q(\mathfrak{g}-1)+1 = \mid \cX_u \mid - \frac{1}{2}q^{2r-1} + \text{lower terms}$. As $k_{ev} = \mid \cX_u \mid-q^{r-1}$ in Equation~\ref{22.09.11}, we can see that there are values of $q$ and $r$ for which $k_{eq} > k_{AG}.$
\end{example}
\begin{example}
Taking $u=\frac{q^r-1}{q-1}$, we have that $\mathfrak{g}:=\frac{(u-1)(q^{r-1}-1)}{2}=\frac{\left(\frac{q^r-1}{q-1}-1\right)(q^{r-1}-1)}{2}$. From Equation~\ref{22.09.12}, $k_{AG}=\mid \cX_u \mid - q(\mathfrak{g}-1)+1 = \mid \cX_u \mid - \frac{1}{2}q^{2r-1} + \text{lower terms}$. As $k_{ev} = \mid \cX_u \mid-q^{r-1}$ in Equation~\ref{22.09.11}, we can see that there are values of $q$ and $r$ for which $k_{eq} > k_{AG}.$
\end{example}

We close this section by finding the maximum rate that a monomial decreasing norm-trace code ${\rm ev}(\mathcal{M})$ would have when the repair scheme of Theorem~\ref{22.06.08} can be applied. As  $\mid \cX_u \mid = ((q - 1)u + 1)q^{r-1}$, we can see that we can repair an erasure on a monomial decreasing norm-trace code ${\rm ev}(\mathcal{M})$ when $\mathcal{M} \subseteq \Delta\left(x^{(q-1)u}, y^{q^{r-1}}\right)$. Thus, we have the following bound for the rate of the code:
\begin{equation}\label{22.09.10}
\text{Rate}({\rm ev}(\mathcal{M})) \leq \frac{(q - 1)uq^{r-1}}{((q - 1)u + 1)q^{r-1}}=1-\frac{1}{(q - 1)u + 1},
\end{equation}
where the inequality is tight when $\mathcal{M} = \Delta\left(x^{(q-1)u}, y^{q^{r-1}}\right)$. In the particular case where $u=\frac{q^r-1}{q-1}$, the inequality in~\ref{22.09.10} becomes:
\[\text{Rate}({\rm ev}(\mathcal{M})) \leq 1-\frac{1}{q^r}.\]

\section*{Conclusion} \label{conclusion}
This work focuses on decreasing norm-trace codes, which are evaluation codes defined by a set of monomials closed under divisibility and the rational points of the extended norm-trace curve. We used Gr\"obner basis theory and indicator functions to find the basic parameters of these codes: length, dimension, minimum distance, and dual code. By exploiting the basic parameters, we gave conditions over the set of monomials, so a decreasing norm-trace code is a self-orthogonal or a self-dual code. We presented a repair scheme for a single erasure on a decreasing norm-trace code that repairs codes with higher rates than the AG codes over the norm-trace curve.

\bibliography{glm_bib_2}{}
\bibliographystyle{abbrv}

\end{document}